\documentclass[11pt,notitlepage]{article}%
\usepackage{amsfonts}
\usepackage{amsmath}
\usepackage{amssymb}
\usepackage{graphicx}%
\providecommand{\U}[1]{\protect\rule{.1in}{.1in}}
\newtheorem{theorem}{Theorem}
\newtheorem{acknowledgement}[theorem]{Acknowledgement}

\newtheorem{lemma}[theorem]{Lemma}

\newtheorem{proposition}[theorem]{Proposition}
\newtheorem{remark}[theorem]{Remark}

\newenvironment{proof}[1][Proof]{\noindent\textbf{#1.} }{\ \rule{0.5em}{0.5em}}
\begin{document}

\title{Preferred Quantization Rules: Born--Jordan vs. Weyl: the Pseudo-Differential
Point of View}
\author{Maurice de Gosson\thanks{Financed by the Austrian Research Agency FWF
(Projektnummer P20442-N13).}\\\textit{Universit\"{a}t Wien, NuHAG}\\\textit{Fakult\"{a}t f\"{u}r Mathematik }\\\textit{A-1090 Wien}
\and Franz Luef\thanks{Financed by the Marie Curie Outgoing Fellowship PIOF
220464.}\\\textit{ University of California}\\\textit{Department of Mathematics}\\\textit{Berkeley CA 94720-3840}}
\maketitle

\begin{abstract}
There has recently been evidence for replacing the usual Weyl quantization
procedure by the older and much less known Born--Jordan rule. In this paper we
discuss this quantization procedure in detail and relate it to recent results
of Boggiato, De Donno, and Oliaro on the Cohen class. We begin with a
discussion of some properties of Shubin's $\tau$-pseudo-differential calculus,
which allows us to show that the Born--Jordan quantization of a symbol $a$ is
the average for $\tau\in\lbrack0,1]$ of the $\tau$-operators with symbol $a$.
We study the properties of the Born--Jordan operators, including their
symplectic covariance, and give their Weyl symbol.

\end{abstract}

\section*{Introduction}

\subsection*{Physical background}

Already in the early days of quantum mechanics physicists were confronted with
the ordering problem for products of observables (i.e. of symbols, in
mathematical language). While it was agreed that the correspondence rule
$x_{j}\longrightarrow x_{j}$, $p_{j}\longrightarrow-i\hbar\partial/\partial
x_{j}$ could be successfully be applied to the position and momentum
variables, thus turning the Hamiltonian function%
\begin{equation}
H=\sum_{j=1}^{N}\frac{1}{2m_{j}}p_{j}^{2}+V(x_{1},..,x_{N})\label{h1}%
\end{equation}
into the partial differential operator%
\begin{equation}
\widehat{H}=\sum_{j=1}^{N}\frac{-\hbar^{2}}{2m_{j}}\frac{\partial^{2}%
}{\partial x_{j}^{2}}+V(x_{1},..,x_{N})\label{h2}%
\end{equation}
it quickly became apparent that these rules lead to fundamental ambiguities
when applied to more general observables involving products of function of
$x_{j}$ and $p_{j}$. For instance, what should the operator corresponding to
the magnetic Hamiltonian%
\begin{equation}
H=\sum_{j=1}^{N}\frac{1}{2m_{j}}\left(  p_{j}-\mathcal{A}_{j}(x_{1}%
,..,x_{N})\right)  ^{2}+V(x_{1},..,x_{N})\label{h3}%
\end{equation}
be? Even if the simple case of the product $x_{j}p_{j}=p_{j}x_{j}$ the
correspondence rule led to the a priori equally good answers $-i\hbar
x_{j}\partial/\partial x_{j}$ and $-i\hbar(\partial/\partial x_{j})x_{j}$
which differ by the quantity $i\hbar$; things became even more complicated
when one came (empirically) to the conclusion that the right answer should in
fact be the \textquotedblleft average rule\textquotedblright\
\begin{equation}
x_{j}p_{j}\longrightarrow-\tfrac{1}{2}i\hbar\left(  x_{j}\tfrac{\partial
}{\partial x_{j}}+\tfrac{\partial}{\partial x_{j}}x_{j}\right)  \label{w1}%
\end{equation}
corresponding to the splitting $x_{j}p_{j}=\frac{1}{2}(x_{j}p_{j}+p_{j}x_{j}%
)$. In 1926 Born and Jordan \cite{bj} proposed to more generally quantize
monomials $x_{j}^{m}p_{j}^{n}$ using the rule%
\begin{equation}
\text{(BJ) \ \ \ }x_{j}^{m}p_{j}^{n}\longrightarrow\frac{1}{n+1}\sum_{k=0}%
^{n}\widehat{p}_{j}^{n-k}\widehat{x}_{j}^{m}\widehat{p}_{j}^{k}\label{bj1}%
\end{equation}
where $\widehat{x}_{j}=$ \textit{multiplication by} $x_{j}$ and $\widehat{p}%
_{j}=i\hbar\partial/\partial x_{j}$ (see Fedak and Prentis \cite{fepr09} for a
readable analysis cast in a \textquotedblleft modern\textquotedblright%
\ language of Born and Jordan's argument; the older papers \cite{ca78} by
Castellani and \cite{cr89} by Crehan also contain valuable information). Born
and Jordan's rules (\ref{bj1}) were actually soon superseded (at least in the
mathematical literature) by Weyl's quantization procedure: in his mathematical
study of quantum mechanics, Weyl proposed in \cite{Weyl} a very general
quantization rule which leads, for monomials, to the replacement of the
Born--Jordan prescription (\ref{bj1}) by
\begin{equation}
\text{(Weyl) \ \ \ }x_{j}^{m}p_{j}^{n}\longrightarrow\frac{1}{2^{n}}\sum
_{k=0}^{n}%
\begin{pmatrix}
n\\
k
\end{pmatrix}
\widehat{p}_{j}^{n-k}\widehat{x}_{j}^{m}\widehat{p}_{j}^{k}.\label{w2}%
\end{equation}
Weyl's rule (which coincides with the Born--Jordan rule when $m+n=2$) nowadays
plays an important role in mathematical analysis (the theory of
pseudo-differential operators), and in physics it has become the preferred
quantization scheme. This is mainly due to two reasons: first to real
observables (or symbols as they are called in mathematics) correspond
(formally) self-adjoint operators; this is a very desirable properties since a
thumb rule in quantum mechanics is that to a real observable should correspond
an operator with real eigenvalues (which are, in quantum mechanics, the values
that the observable can actually take). Another advantage of the Weyl
correspondence is of a more subtle nature: it is the symplectic covariance
property. This property which is actually \emph{characteristic} of the Weyl
correspondence among all other pseudo-differential calculi (Wong \cite{Wong})
says that if we perform a linear symplectic change of variables in the symbol,
then the resulting operator is conjugated to the original by a certain unitary
operator obtained by the metaplectic representation. A third property, which
is in a sense rather unwelcome (Kauffmann \cite{kauffmann}) is that the Weyl
correspondence is invertible (see e.g. Wong \cite{Wong}). Invertibility poses
severe epistemological problems, because it is not physically founded. It is
actually possible to prove (de Gosson and Hiley \cite{gohi}) that there is a
one-to-one correspondence between Hamiltonian flows and the continuous groups
of operators in $L^{2}(\mathbb{R}^{N})$ solving the Schr\"{o}dinger equation
with Hamiltonian obtained by Weyl quantization. This result in a sense
\textquotedblleft trivializes\textquotedblright\ quantum mechanics making it
appear merely as a \textquotedblleft copy\textquotedblright\ of Hamiltonian
mechanics. This issue, which is related to \textquotedblleft
dequantization\textquotedblright, will be briefly discussed at the end of the
present paper.

\subsection*{Aims and structure of the paper}

Shubin's $\tau$-pseudo-differential calculus (which we review and complement
in Section \ref{sec1}) suggests to consider variants of the usual Wigner
distribution of the type%
\[
W_{\tau}(\psi,\phi)(z)=\left(  \tfrac{1}{2\pi\hbar}\right)  ^{N}%
\int_{\mathbb{R}^{N}}e^{-\frac{i}{\hbar}py}\psi(x+\tau y)\overline{\phi
}(x-(1-\tau)y)dy
\]
where $\tau$ is a real parameter (the choice $\tau=\frac{1}{2}$ yields the
usual cross-Wigner distribution). Recently Boggiatto et al \cite{bogetal}
(also see Boggiatto et al \cite{bogetalbis}) have shown the advantages of
using the average%
\[
Q(\psi,\phi)=\int_{0}^{1}W_{\tau}(\psi,\phi)d\tau
\]
of these $\tau$-distributions on the interval $[0,1]$. Besides the fact that
it belongs to the Cohen class and has the right marginals (which is an
essential feature in quantum mechanics), the distribution $Q(\psi,\phi)$
almost entirely eliminates the interference phenomenon (\textquotedblleft
ghost frequencies\textquotedblright) presented by the distributions $W_{\tau
}(\psi,\phi)$. This property makes of $Q(\psi,\phi)$ a tool of choice in
time-frequency analysis. Recalling that the Wigner transform $W(\psi
,\phi)=W_{1/2}(\psi,\phi)$ is related to the Weyl operator $\widehat{A}%
=\operatorname*{Op}(a)$ by the formula%
\[
(\widehat{A}\psi|\phi)_{L^{2}}=\langle a,W(\psi,\phi)\rangle
\]
this suggests to define a new type of pseudo-differential operator
$\widetilde{A}$ by the formula%
\[
(\widetilde{A}\psi|\phi)_{L^{2}}=\langle a,Q(\psi,\phi)\rangle;
\]
not very surprisingly that operator $\widetilde{A}$ is also an
\textquotedblleft average\textquotedblright, namely%
\[
\widetilde{A}=\int_{0}^{1}\widehat{A}_{\tau}d\tau
\]
where $\widehat{A}_{\tau}=\operatorname*{Op}_{\tau}(a)$ is the Shubin $\tau
$-pseudo-differential operator with symbol $a$. We will show in this paper
that this operator $\widetilde{A}$ (which is also studied in Boggiatto et al
\cite{bogetal}) is precisely the \emph{Born--Jordan quantization}
$\widehat{A}_{\mathrm{BJ}}$ of the symbol $a$ (see Section \ref{sec3}).We will
show that Born--Jordan quantization allows to recover the rules (\ref{bj1})
when the symbol is a monomial. We will also prove in Proposition
\ref{propharmonic} a harmonic decomposition result for the operator
$\widetilde{A}=\widehat{A}_{\mathrm{BJ}}$, namely
\[
\widehat{A}_{\mathrm{BJ}}\psi=\left(  \tfrac{1}{2\pi\hbar}\right)  ^{N}%
\int_{\mathbb{R}^{2N}}\mathcal{F}_{\sigma}a(z_{0})\Theta(z_{0})\widehat{T}%
(z_{0})\psi dz_{0}%
\]
where $\widehat{T}(z_{0})$ is the usual Heisenberg operator,\ $\mathcal{F}%
_{\sigma}a$ the symplectic Fourier transform of the symbol, and $\Theta$ is
the function defined by%
\[
\Theta(z_{0})=\frac{\sin(p_{0}x_{0}/\hbar)}{p_{0}x_{0}/\hbar}%
\]
which also appears (for $\hbar=1/2\pi$) in the work of Boggiatto et al
\cite{bogetal}; the formula above shows, in particular, that the Weyl symbol
of $\widehat{A}_{\mathrm{BJ}}$ is given by the formula
\[
a_{\mathrm{W}}=\left(  \tfrac{1}{2\pi\hbar}\right)  ^{N}a\ast\mathcal{F}%
_{\sigma}\Theta.
\]
We also discuss the symplectic covariance properties of the Born--Jordan
quantization; we prove that this covariance holds for an interesting subgroup
of the metaplectic group, namely the group generated by the metalinear group
and the Fourier transform (full symplectic covariance cannot of course be
expected since the latter is characteristic of Weyl quantization as has been
shown in detail by Wong \cite{Wong}).

\subsection*{Notation}

We will write $x=(x_{1},..,x_{N})$, $p=(p_{1},..,p_{N})$, and $z=(x,p)$.
Scalar products will be denoted $xx^{\prime}$, $pp^{\prime}$, etc. For
instance $px=p_{1}x_{1}+\cdot\cdot\cdot+p_{N}x_{N}$. The standard symplectic
form on $\mathbb{R}^{2N}\equiv\mathbb{R}^{N}\oplus\mathbb{R}^{N}$ is given by
$\sigma(z,z^{\prime})=px^{\prime}-p^{\prime}x$. The associated symplectic
group is denoted $\operatorname*{Sp}(2N,\mathbb{R})$. We use the notation
$\mathcal{F}$ for the $\hbar$-dependent unitary Fourier transform:%
\[
\mathcal{F}\psi(p)=\left(  \tfrac{1}{2\pi\hbar}\right)  ^{N/2}\int%
_{\mathbb{R}^{N}}e^{ipx}\psi(x)dx.
\]
The scalar product of two functions $\psi,\phi$ on $\mathbb{R}^{N}$ is
$(\psi|\phi)_{L^{2}}=\int_{\mathbb{R}^{N}}\psi(x)\overline{\phi}(x)dx$ and the
associated norm is denoted by $||\psi||_{L^{2}}$.

\section{Pseudo-Differential Operators\label{sec1}}

\subsection{Definitions and first properties}

\subsubsection{The operators $\protect\widehat{A}_{\tau}$}

The consideration of different quantization rules leads us to study
pseudo-differential operators of the type%
\begin{equation}
\widehat{A}_{\tau}\psi(x)=\left(  \tfrac{1}{2\pi\hbar}\right)  ^{N}%
\int_{\mathbb{R}^{2N}}e^{\frac{i}{\hbar}p\cdot(x-y)}a(\tau x+(1-\tau
)y,p)\psi(y)dydp \label{aftau}%
\end{equation}
where $\tau$ is\ a real parameter (Shubin \cite{sh87}); the integral should be
understood in some \textquotedblleft reasonable\textquotedblright\ sense, see
below. We will often use the notation
\[
\widehat{A}_{\tau}=\operatorname*{Op}\nolimits_{\tau}(a).
\]
For instance, if $\psi\in\mathcal{S}(\mathbb{R}^{N})$ and $a\in\mathcal{S}%
(\mathbb{R}^{2N})$ the integral is absolutely convergent. For more general
symbols $a$ (for instance $a\in\mathcal{S}^{\prime}(\mathbb{R}^{2N})$) one can
give a meaning to the expression (\ref{aftau}) by declaring that the operator
$\widehat{A}$ is defined by the distributional kernel
\begin{equation}
K_{\tau}(x,y)=\left(  \tfrac{1}{2\pi\hbar}\right)  ^{N/2}(\mathcal{F}_{2}%
^{-1}a)((1-\tau)x+\tau y,p) \label{ker1}%
\end{equation}
where $\mathcal{F}_{2}^{-1}$ is the inverse Fourier transform in the second
set of variables; it is however more natural in our context to use the method
explained after Proposition \ref{propos1} below, and which makes use of the
$\tau$-Wigner transform. We notice that setting $\tau=\frac{1}{2}$ in formula
(\ref{aftau}) we recover the expression%
\begin{equation}
\widehat{A}\psi(x)=\left(  \tfrac{1}{2\pi\hbar}\right)  ^{N}\int%
_{\mathbb{R}^{2N}}e^{\frac{i}{\hbar}p(x-y)}a(\tfrac{1}{2}(x+y),p)\psi(y)dydp.
\label{w3}%
\end{equation}
of a Weyl operator which is standard in the theory of partial differential
operators. When $\tau=1$ formula (\ref{aftau}) can be rewritten
\begin{equation}
\widehat{A}\psi(x)=\left(  \tfrac{1}{2\pi\hbar}\right)  ^{N/2}\int%
_{\mathbb{R}^{N}}e^{\frac{i}{\hbar}px}a(x,p)\mathcal{F}\psi(p)dp
\label{conpdo}%
\end{equation}
where $\mathcal{F}\psi$ is the Fourier transform of $\psi$; this is the
conventional definition of a pseudo-differential operator found in most texts
dealing with partial differential equations and $a$ is then sometimes called
the \textquotedblleft Kohn--Nirenberg symbol\textquotedblright\ of the
operator $\widehat{A}$. The Kohn--Nirenberg calculus is used mainly in the
microlocal analysis of partial differential equations, and in time-frequency
analysis where it is sometimes more tractable for computational purposes than
the Weyl correspondence. One immediately checks that Kohn--Nirenberg operators
correspond to the simple ordering rule%
\begin{equation}
\text{(KN) \ \ }x_{j}^{m}p_{j}^{n}\longrightarrow\widehat{x}_{j}%
^{m}\widehat{p}_{j}^{n} \label{w4}%
\end{equation}
in the case of monomials.

A well-known property of the Weyl operators $\widehat{A}=\operatorname*{Op}%
_{1/2}(a)$ is that the (formal) adjoint is given by $\widehat{A}^{\ast
}=\operatorname*{Op}_{1/2}(\overline{a})$; in the $\tau$-dependent case we
have the more general relation%
\begin{equation}
\operatorname*{Op}\nolimits_{\tau}(a)^{\ast}=\operatorname*{Op}%
\nolimits_{1-\tau}(\overline{a}) \label{adjtau}%
\end{equation}
valid for every real $\tau$.

\subsubsection{The quasi-distribution $W_{\tau}$}

An associated object is the $\tau$-Wigner distribution; it is defined as
follows: for a pair $(\psi,\phi)$ of functions in $\mathcal{S}(\mathbb{R}%
^{N})$ one sets%
\begin{equation}
W_{\tau}(\psi,\phi)(z)=\left(  \tfrac{1}{2\pi\hbar}\right)  ^{N}%
\int_{\mathbb{R}^{N}}e^{-\frac{i}{\hbar}py}\psi(x+\tau y)\overline{\phi
}(x-(1-\tau)y)dy\label{wtau}%
\end{equation}
As is the case for $W$ the mapping $W_{\tau}$ is a bilinear and continuous
mapping $\mathcal{S}(\mathbb{R}^{N})\times\mathcal{S}(\mathbb{R}%
^{N})\longrightarrow\mathcal{S}(\mathbb{R}^{2N})$. When $\psi=\phi$ one writes
$W_{\tau}(\psi,\psi)=W_{\tau}\psi$. Of course, when $\tau=\frac{1}{2}$ one
recovers the usual cross-Wigner transform
\begin{equation}
W(\psi,\phi)(z)=\left(  \tfrac{1}{2\pi\hbar}\right)  ^{N}\int_{\mathbb{R}^{N}%
}e^{-\frac{i}{\hbar}py}\psi(x+\tfrac{1}{2}y)\overline{\phi}(x-\tfrac{1}%
{2}y)dy.\label{wigo1}%
\end{equation}
If $\tau=0$ we get%
\[
W_{0}(\psi,\phi)(z)=\left(  \tfrac{1}{2\pi\hbar}\right)  ^{N/2}e^{-\frac
{i}{\hbar}px}\psi(x)\overline{\mathcal{F}\phi}(p)
\]
hence $W_{0}(\psi,\phi)$ is the Rihaczek--Kirkwood distribution $R(\psi,\phi)$
well-known from time-frequency analysis (Gr\"{o}chenig \cite{Gro}, Boggiatto
et al. \cite{bogetal}); if $\tau=1$ one gets the so-called dual
Rihaczek--Kirkwood distribution $R^{\ast}(\phi,\psi)$. It is easily verified
that%
\[
W_{\tau}(\phi,\psi)=\overline{W_{1-\tau}(\psi,\phi)}.
\]
The distribution $W_{\tau}\psi=W_{\tau}(\phi,\psi)$ satisfies the usual
marginal properties:%
\begin{equation}
\int_{\mathbb{R}^{N}}W_{\tau}\psi(z)dp=|\psi(x)|^{2}\text{ ,\ }\int%
_{\mathbb{R}^{N}}W_{\tau}\psi(z)dx=|\mathcal{F}\psi(p)|^{2}\text{\ }%
\label{marginal1}%
\end{equation}
(see Boggiatto et al. \cite{bogetal})

There is a fundamental relation between Weyl pseudo-differential operators and
the cross-Wigner transform, that relation is often used to define the Weyl
operator $\widehat{A}=\operatorname*{Op}(a)$:
\begin{equation}
(\operatorname*{Op}(a)\psi|\phi)_{L^{2}}=\langle a,W(\psi,\phi)\rangle
\label{aweyl}%
\end{equation}
for $\psi,\phi\in\mathcal{S}(\mathbb{R}^{N})$. Not very surprisingly this
formula extends to the case of $\tau$-operators:

\begin{proposition}
\label{propos1}Let $\psi,\phi\in\mathcal{S}(\mathbb{R}^{N})$, $a\in
\mathcal{S}(\mathbb{R}^{2N})$, and $\tau$ a real number. We have%
\begin{equation}
(\operatorname*{Op}\nolimits_{\tau}(a)\psi|\phi)_{L^{2}}=\langle a,W_{\tau
}(\psi,\phi)\rangle\label{awtau}%
\end{equation}
where $\langle\cdot,\cdot\rangle$ is the distributional bracket on
$\mathbb{R}^{2N}$.
\end{proposition}

\begin{proof}
By definition of $W_{\tau}$ we have%
\begin{multline*}
\langle a,W_{\tau}(\psi,\phi)\rangle=\\
\left(  \tfrac{1}{2\pi\hbar}\right)  ^{N}\int_{\mathbb{R}^{3N}}e^{-\frac
{i}{\hbar}p\cdot y}a(z)\psi(x+\tau y)\overline{\phi}(x-(1-\tau)y)dydpdx
\end{multline*}
and setting $x+\tau y=y^{\prime}$, $x-(1-\tau)y=y^{\prime}$ this is
\begin{multline*}
\langle a,W_{\tau}(\psi,\phi)\rangle=\\
\left(  \tfrac{1}{2\pi\hbar}\right)  ^{N}\int_{\mathbb{R}^{3N}}e^{-\frac
{i}{\hbar}p\cdot(x^{\prime}-y^{\prime})}a((1-\tau)x^{\prime}+\tau y^{\prime
},p)\psi(y^{\prime})\overline{\phi}(x^{\prime})dydpdx
\end{multline*}
hence the equality (\ref{awtau}) in view of definition (\ref{aftau}) of the
operator $\widehat{A}_{\tau}=\operatorname*{Op}\nolimits_{\tau}(a)$.
\end{proof}

Formula (\ref{awtau}) allows us to define $\widehat{A}_{\tau}\psi
=\operatorname*{Op}\nolimits_{\tau}(a)\psi$ for arbitrary symbols
$a\in\mathcal{S}^{\prime}(\mathbb{R}^{2N})$ and $\psi\in\mathcal{S}%
(\mathbb{R}^{N})$ in the same way as is done for Weyl pseudo-differential
operators: choose $\phi\in\mathcal{S}(\mathbb{R}^{N})$; then $W_{\tau}%
(\psi,\phi)\in\mathcal{S}(\mathbb{R}^{2N})$ and the distributional bracket
$\langle a,W_{\tau}(\psi,\phi)\rangle$ is thus well-defined; by definition
$\widehat{A}_{\tau}\psi$ is given by (\ref{awtau}), and $\widehat{A}_{\tau}$
is a continuous operator $\mathcal{S}(\mathbb{R}^{N})\longrightarrow
\mathcal{S}^{\prime}(\mathbb{R}^{N})$.

\begin{remark}
\label{rem1}It follows from the argument above and using Schwartz's kernel
theorem that every continuous operator $\mathcal{S}(\mathbb{R}^{N}%
)\longrightarrow\mathcal{S}^{\prime}(\mathbb{R}^{N})$ is an operator of the
type $\widehat{A}_{\tau}$ for every value of the parameter $\tau$; the
argument goes exactly as in the standard case of Weyl operators treated in
Shubin \cite{sh87} or Gr\"{o}chenig \cite{Gro}.
\end{remark}

In Weyl calculus the introduction of the Wigner transform $W\psi$ of a square
integrable function has the following very simple and natural interpretation:
it is, up to a constant factor, the Weyl symbol of the projection operator
$\Pi_{\psi}$ of $L^{2}(\mathbb{R}^{N})$ on the ray $\{\lambda\psi:\lambda
\in\mathbb{C}\}$. This interpretation extends to the $\tau$-dependent case
without difficulty:

\begin{proposition}
Let $\psi\in L^{2}(\mathbb{R}^{N})$.

(i) We have $\Pi_{\psi}=\left(  2\pi\hbar\right)  ^{N}\operatorname*{Op}%
\nolimits_{\tau}(W_{\tau}\psi)$;

(ii) The $\tau$-symbol of the operator with kernel $K=\psi\otimes
\overline{\phi}$ is $\left(  2\pi\hbar\right)  ^{N}W_{\tau}(\psi,\phi).$
\end{proposition}

\begin{proof}
(i) Let $\phi\in L^{2}(\mathbb{R}^{N})$; by definition $\Pi_{\psi}\phi
=(\phi|\psi)_{L^{2}}\psi$ that is
\[
\Pi_{\psi}\phi(x)=\int_{\mathbb{R}^{N}}\psi(x)\overline{\psi}(y)\phi(y)dy
\]
hence the kernel of $\Pi_{\psi}$ is $K(x,y)=\psi(x)\overline{\psi}(y)$. Using
a partial Fourier inversion formula, formula (\ref{ker1}) expressing the
kernel of $\Pi_{\psi}$ in terms of its $\tau$-symbol $\pi_{\psi}$ can be
rewritten%
\begin{align*}
\pi_{\psi}(x,p)  &  =\int_{\mathbb{R}^{N}}e^{-\frac{i}{\hbar}py}K(x+\tau
y,x-(1-\tau)y)dy\\
&  =\int_{\mathbb{R}^{N}}e^{-\frac{i}{\hbar}py}\psi(x+\tau y)\overline{\psi
}(x-(1-\tau)y)dy\\
&  =\left(  2\pi\hbar\right)  ^{N}W_{\tau}\psi(x,p).
\end{align*}
The assertion (ii) is proven in a similar way replacing $\psi\otimes
\overline{\psi}$ with $\psi\otimes\overline{\phi}$ in the argument above.
\end{proof}

We also have the Moyal identity:

\begin{proposition}
Let $((\cdot|\cdot))_{L^{2}}$ be the scalar product on $L^{2}(\mathbb{R}%
^{2N})$ and $|||\cdot|||_{L^{2}}$ the associated norm. We have
(\textquotedblleft Moyal identity\textquotedblright)%
\begin{equation}
((W_{\tau}(\psi,\phi)|W_{\tau}(\psi^{\prime},\phi^{\prime})))_{L^{2}}=\left(
\tfrac{1}{2\pi\hbar}\right)  ^{N}(\psi|\psi^{\prime})_{L^{2}}\overline
{(\phi|\phi^{\prime})_{L^{2}}} \label{moy1}%
\end{equation}
and hence in particular%
\begin{equation}
|||W_{\tau}(\psi,\phi)|||_{L^{2}}=\left(  \tfrac{1}{2\pi\hbar}\right)
^{N/2}||\psi||_{L^{2}}||\phi||_{L^{2}} \label{moy2}%
\end{equation}
for all $\psi,\psi^{\prime},\phi,\phi^{\prime}\in L^{2}(\mathbb{R}^{N})$.
\end{proposition}

\begin{proof}
Let us set
\[
I=\left(  2\pi\hbar\right)  ^{2N}((W_{\tau}(\psi,\phi)|W_{\tau}(\psi^{\prime
},\phi^{\prime})))_{L^{2}}.
\]
We have, by definition of $W_{\tau}$,
\begin{multline}
I=\int_{\mathbb{R}^{4N}}e^{-\frac{i}{\hbar}p(y-y^{\prime})}\\
\times\psi(x+\tau y)\psi^{\prime}(x+\tau y)\overline{\phi}(x-(1-\tau
)y^{\prime})\overline{\phi^{\prime}}(x-(1-\tau)y^{\prime})dxdpdydy^{\prime
}.\nonumber
\end{multline}
The integral in $p$ is $\left(  2\pi\hbar\right)  ^{N}\delta(y-y^{\prime})$
hence%
\begin{multline*}
I=\left(  2\pi\hbar\right)  ^{N}\\
\times\int_{\mathbb{R}^{2N}}\psi(x+\tau y)\psi^{\prime}(x+\tau y)\overline
{\phi}(x-(1-\tau)y)\overline{\phi^{\prime}}(x-(1-\tau)y)dxdy.
\end{multline*}
Setting $u=x+\tau y$ and $v=x-(1-\tau)y$ we have $dudv=dxdy$ and hence%
\begin{align*}
I  &  =\left(  2\pi\hbar\right)  ^{N}\int_{\mathbb{R}^{2N}}\psi(u)\psi
^{\prime}(u)\overline{\phi}(v)\overline{\phi^{\prime}}(v)dudv\\
&  =\left(  2\pi\hbar\right)  ^{N}(\psi|\psi^{\prime})_{L^{2}}\overline
{(\phi|\phi)_{L^{2}}}%
\end{align*}
which proves (\ref{moy1}); formula (\ref{moy2}) follows.
\end{proof}

\subsubsection{Ordering of monomials}

Since we are dealing in this paper with ordering issues let us find the $\tau
$-pseudo-differential operator corresponding to the monomial symbols
$x_{j}^{m}p_{j}^{n}$ considered in the Introduction:

\begin{proposition}
\label{propos5}Let $m$ and $n$ be two non-negative integers. We have
\begin{equation}
\operatorname*{Op}\nolimits_{\tau}(x_{j}^{m}p_{j}^{n})=\sum_{k=0}^{m}%
\begin{pmatrix}
m\\
k
\end{pmatrix}
\tau^{k}(1-\tau)^{m-k}\widehat{x_{j}}^{k}\widehat{p_{j}}^{n}\widehat{x_{j}%
}^{m-k} \label{txp1}%
\end{equation}
or, equivalently,
\begin{equation}
\operatorname*{Op}\nolimits_{\tau}(x_{j}^{m}p_{j}^{n})=\sum_{k=0}^{n}%
\begin{pmatrix}
n\\
k
\end{pmatrix}
(1-\tau)^{k}\tau^{n-k}\widehat{p_{j}}^{k}\widehat{x_{j}}^{n}\widehat{p_{j}%
}^{n-k} \label{txp2}%
\end{equation}
where $\widehat{x_{j}}^{\ell}\psi=x_{j}^{\ell}\psi$ and $\widehat{p_{j}}%
^{\ell}\psi=(-i\hbar\partial_{x_{j}})^{\ell}\psi$.
\end{proposition}

\begin{proof}
It is sufficient to assume $N=1$ so we write $x_{j}^{m}=x^{m}$ and $p_{j}%
^{n}=p^{n}$. Let us set $a_{m,n}(z)=x^{m}p^{n};$ we have using the binomial
formula%
\begin{equation}
a_{m,n}(\tau x+(1-\tau)y,p)=\sum_{k=0}^{m}%
\begin{pmatrix}
m\\
k
\end{pmatrix}
\tau^{k}(1-\tau)^{m-k}x^{k}y^{m-k}p^{n}. \label{amn}%
\end{equation}
Setting $b_{m,n,k}(z)=x^{k}y^{m-k}p^{n}$ we have (in the sense of
distributions)
\[
\operatorname*{Op}\nolimits_{\tau}(b_{m,n,k})\psi(x)=\tfrac{1}{2\pi\hbar}%
x^{k}\int_{-\infty}^{\infty}\left[  \int_{-\infty}^{\infty}e^{\frac{i}{\hbar
}p(x-y)}p^{n}dp\right]  y^{m-k}\psi(y)dy.
\]
Using the Fourier inversion formula
\begin{equation}
\tfrac{1}{2\pi\hbar}\int_{-\infty}^{\infty}e^{\frac{i}{\hbar}p(x-y)}%
p^{n}dp=(-i\hbar)^{n}\delta^{(n)}(x-y) \label{fif}%
\end{equation}
we thus have%
\[
\operatorname*{Op}\nolimits_{\tau}(b_{m,n,k})\psi=x^{k}(-i\hbar)^{n}%
\partial_{x}^{n}(x^{m-k}\psi).
\]
Formula (\ref{txp1}) follows inserting this expression in (\ref{amn}). To
prove that this formula is equivalent to (\ref{txp2}) the easiest method
consists in remarking that we have the conjugation formula%
\[
\mathcal{F}\operatorname*{Op}\nolimits_{\tau}(a)\mathcal{F}^{-1}%
=\operatorname*{Op}\nolimits_{1-\tau}(a\circ J^{-1})
\]
(which will be proven in Proposition \ref{propos6} below). Since we have
$a_{m,n}(J^{-1}z)=(-1)^{m}x^{m}p^{n}$ and, using the standard properties of
the Fourier transform,
\[
\mathcal{F}\operatorname*{Op}\nolimits_{\tau}(a_{m,n})\mathcal{F}%
^{-1}=(-1)^{m}\widehat{p}^{k}\widehat{x}^{n}\widehat{p}^{m-k}%
\]
formula (\ref{txp2}) follows.
\end{proof}

\begin{remark}
Taking $\tau=\frac{1}{2}$ in either formula (\ref{txp1}) or (\ref{txp2}) we
recover Weyl's ordering rule (\ref{w2}). Similarly, taking $\tau=0$, one gets
the Kohn--Nirenberg ordering rule (\ref{w4}).
\end{remark}

\subsection{Symplectic covariance properties}

\subsubsection{Conjugation with Fourier transform}

As already mentioned in the Introduction a characteristic property of Weyl
quantization is symplectic covariance. This property can be described as
follows: let $\operatorname*{Sp}(2N,\mathbb{R})$ be the standard symplectic
group of $\mathbb{R}^{2N}$: it is the group of linear automorphisms $s$ of
$\mathbb{R}^{2N}$ such that $s^{T}Js=J$ where $J$ is the matrix $%
\begin{pmatrix}
0_{N} & I_{N}\\
-I_{N} & 0_{N}%
\end{pmatrix}
$; equivalently $s\in\operatorname*{Sp}(2N,\mathbb{R})$ if and only if
$\sigma(sz,sz^{\prime})=\sigma(z,z^{\prime})$ for all $z,z^{\prime}%
\in\mathbb{R}^{2N}$ where $\sigma(z,z^{\prime})=Jz\cdot z^{\prime}$ is the
standard symplectic form. The group $\operatorname*{Sp}(2N,\mathbb{R})$ is
connected and $\pi_{1}[\operatorname*{Sp}(2N,\mathbb{R})]\equiv(\mathbb{Z},+)$
so that it has a connected covering group $\operatorname*{Sp}_{2}%
(2N,\mathbb{R})$ of order $2$. That group has a faithful representation by a
group of unitary operators on $L^{2}(\mathbb{R}^{N})$, the metaplectic group
$\operatorname*{Mp}(2N,\mathbb{R})$. Let $\pi^{\operatorname*{Mp}%
}:\operatorname*{Mp}(2N,\mathbb{R})\longrightarrow\operatorname*{Sp}%
(2N,\mathbb{R})$ be the natural projection; to every $s\in\operatorname*{Sp}%
(2N,\mathbb{R})$ thus correspond two elements $\pm S$ of $\operatorname*{Mp}%
(2N,\mathbb{R})$ such that $\pi^{\operatorname*{Mp}}(\pm S)=s$. The symplectic
covariance of Weyl calculus means that if $\widehat{A}=\operatorname*{Op}(a)$
then%
\begin{equation}
S^{-1}\widehat{A}S=\operatorname*{Op}(a\circ s). \label{spmp}%
\end{equation}
This property is equivalent to the following property of the cross-Wigner
transform:%
\begin{equation}
W(S\psi,S\phi)(z)=W(\psi,\phi)(s^{-1}z) \label{wigco}%
\end{equation}
(it is an easy exercise to deduce this equivalence from formula (\ref{aweyl}%
)). Property (\ref{spmp}) is \emph{characteristic} of Weyl calculus: let
$\widehat{A}$ be a linear continuous operator $\mathcal{S}(\mathbb{R}%
^{N})\longrightarrow\mathcal{S}^{\prime}(\mathbb{R}^{N})$ and write it as a
$\tau$-operator $\widehat{A}_{\tau}=\operatorname*{Op}_{\tau}(a)$ (formula
(\ref{aftau}); cf. Remark \ref{rem1}). Then if $S^{-1}\widehat{A}S$, again
viewed as a $\tau$-operator, has symbol $a\circ s$ we must have $\tau=\frac
{1}{2}.$ For this reason one cannot expect a general symplectic covariance
property for the $\tau$-pseudo-differential calculus unless $\tau=\frac{1}{2}%
$. For instance Boggiatto et al. prove in \cite{bogetal} that $W_{1-\tau
}(\mathcal{F}\psi)(p,-x)=W_{\tau}\psi(x,p)$ when $\mathcal{F}$ is the Fourier
transform; in fact the same argument shows that, more generally,
\[
W_{1-\tau}(\mathcal{F}\psi,\mathcal{F}\phi)(p,-x)=W_{\tau}(\psi,\phi)(x,p).
\]
Now, the modified Fourier transform $F=e^{-iN\pi/4}\mathcal{F}$ is in
$\operatorname*{Mp}(2N,\mathbb{R})$ and we have precisely $\pi
^{\operatorname*{Mp}}(F)=J$ hence the formula above can be written in a more
symplectic fashion as%
\begin{equation}
W_{1-\tau}(F\psi,F\phi)(z)=W_{\tau}(\psi,\phi)(J^{-1}z) \label{wtauf}%
\end{equation}
which reduces to (\ref{wigco}) in the case $s=J$ if and only if $\tau=\frac
{1}{2}$. Formula (\ref{wtauf}) has the following interesting consequence for
$\tau$-pseudo-differential operators:

\begin{proposition}
\label{propos6}Let $\widehat{A}_{\tau}=\operatorname*{Op}_{\tau}(a)$,
$a\in\mathcal{S}^{\prime}(\mathbb{R}^{2N})$. We have%
\begin{equation}
\mathcal{F}\operatorname*{Op}\nolimits_{\tau}(a)\mathcal{F}^{-1}%
=\operatorname*{Op}\nolimits_{1-\tau}(a\circ J^{-1}) \label{optauf}%
\end{equation}

\end{proposition}

\begin{proof}
Let $\psi,\phi\in\mathcal{S}(\mathbb{R}^{N})$; since $\mathcal{F}$ is unitary
we have%
\[
(\mathcal{F}\operatorname*{Op}\nolimits_{\tau}(a)\mathcal{F}^{-1}\psi
|\phi)_{L^{2}}=(\operatorname*{Op}\nolimits_{\tau}(a)\mathcal{F}^{-1}%
\psi|\mathcal{F}^{-1}\phi)_{L^{2}}%
\]
hence, using twice (\ref{wtauf}),%
\begin{align*}
(\mathcal{F}\operatorname*{Op}\nolimits_{\tau}(a)\mathcal{F}^{-1}\psi
|\phi)_{L^{2}}  &  =\langle a,W_{\tau}(\psi,\phi)\circ J\rangle\\
&  =\langle a\circ J^{-1},W_{1-\tau}(\psi,\phi)\rangle\\
&  =(\operatorname*{Op}\nolimits_{1-\tau}(a\circ J^{-1})\psi|\phi)_{L^{2}}%
\end{align*}
which implies (\ref{optauf}) since $\psi$ and $\phi$ are arbitrary.
\end{proof}

\begin{remark}
Formula (\ref{optauf}) in Proposition \ref{propos6} allows us to give a very
short proof of the fact that Weyl operators are the only pseudo-differential
operators satisfying the property of symplectic covariance (cf. the proof in
Wong \cite{Wong}). Indeed, replacing $\mathcal{F}$ by $F$ in (\ref{optauf}) we
see that $F\operatorname*{Op}\nolimits_{\tau}(a)F^{-1}=\operatorname*{Op}%
\nolimits_{\tau}(a\circ J^{-1})$ if and only if $\tau=\frac{1}{2}$. One
concludes by noting that $F\in\operatorname*{Mp}(2N,\mathbb{R})$.
\end{remark}

\subsubsection{Covariance under the metalinear group}

However, symplectic covariance subsists for an important subgroup of the
metaplectic group $\operatorname*{Mp}(2N,\mathbb{R})$. Let $m_{L}$ be the
automorphism of $\mathbb{R}^{2N}$ defined, for $L\in\operatorname*{GL}%
(N,\mathbb{R})$, by $m_{L}(x,p)=(L^{-1}x,L^{T}p)$. One immediately verifies
that $m_{L}\in\operatorname*{Sp}(2N,\mathbb{R})$. Moreover, each $m_{L}$ is
the projection onto $\operatorname*{Sp}(2N,\mathbb{R})$ of the two operators
$M_{L,\mu}$ and $M_{L,\mu+2}=-M_{L,\mu}$ in $\operatorname*{Mp}(2N,\mathbb{R}%
)$ defined by
\[
M_{L,\mu}\psi(x)=i^{\mu}\sqrt{|\det L|}\psi(Lx);
\]
here $\mu$ (the \textquotedblleft Maslov index\textquotedblright, see de
Gosson \cite{Birk}) is $0$ or $2$ (modulo $4$) if $\det L>0$ and $1$ or $3$
(modulo $4$) if $\det L<0$. The operators $M_{L,\mu}$ satisfy the
multiplication rule $M_{L,\mu}M_{L^{\prime},\mu^{\prime}}=M_{L^{\prime}%
L,\mu+\mu^{\prime}}$ and thus form a group of unitary operators, the
\textit{metalinear group} $\operatorname*{ML}(2N,\mathbb{R})$.

\begin{proposition}
\label{propos2}Let $M_{L,\mu}\in\operatorname*{ML}(2N,\mathbb{R})$. We have
\begin{equation}
W_{\tau}(M_{L,\mu}\psi,M_{L,\mu}\phi)(z)=W_{\tau}(\psi,\phi)(m_{L}^{-1}z)
\label{metal}%
\end{equation}
for $\psi,\phi\in\mathcal{S}(\mathbb{R}^{N})$ and
\begin{equation}
M_{L,\mu}^{-1}\operatorname*{Op}\nolimits_{\tau}(a)M_{L,\mu}%
=\operatorname*{Op}\nolimits_{\tau}(a\circ m_{L}) \label{aml}%
\end{equation}
for $a\in\mathcal{S}^{\prime}(\mathbb{R}^{2N})$.
\end{proposition}

\begin{proof}
We have%
\begin{multline*}
W_{\tau}(M_{L,\mu}\psi,M_{L,\mu}\phi)(z)=\left(  \tfrac{1}{2\pi\hbar}\right)
^{N}|\det L|\\
\times\int_{\mathbb{R}^{N}}e^{-\frac{i}{\hbar}p\cdot y}\psi(L(x+\tau
y))\overline{\phi}(L(x-(1-\tau)y))dy
\end{multline*}
that is, setting $y^{\prime}=Ly$,%
\begin{multline*}
W_{\tau}(M_{L,\mu}\psi,M_{L,\mu}\phi)(z)=\\
\left(  \tfrac{1}{2\pi\hbar}\right)  ^{N}\int_{\mathbb{R}^{N}}e^{-\frac
{i}{\hbar}p\cdot L^{-1}y^{\prime}}\psi(Lx+\tau y^{\prime})\overline{\phi
}(Lx-(1-\tau)y^{\prime}))dy^{\prime}%
\end{multline*}
hence (\ref{metal}). To prove formula (\ref{aml}) we begin by noting that
\[
(M_{L,\mu}^{-1}\operatorname*{Op}\nolimits_{\tau}(a)M_{L,\mu}\psi|\phi
)_{L^{2}}=(\operatorname*{Op}\nolimits_{\tau}(a)M_{L,\mu}\psi|M_{L,\mu}%
\phi)_{L^{2}}%
\]
that is, using (\ref{awtau}) in Proposition \ref{propos1}, formula
(\ref{metal}), and again formula (\ref{awtau}):
\begin{align*}
(M_{L,\mu}^{-1}\operatorname*{Op}\nolimits_{\tau}(a)M_{L,\mu}\psi|\phi
)_{L^{2}}  &  =\int_{\mathbb{R}^{N}}a(z)W_{\tau}(M_{L,\mu}\psi,M_{L,\mu}%
\phi)(z)dz\\
&  =\int_{\mathbb{R}^{N}}a(z)W_{\tau}(\psi,\phi)(m_{L}^{-1}z)dz\\
&  =\int_{\mathbb{R}^{N}}a(m_{L}z)W_{\tau}(\psi,\phi)(z)dz\\
&  =(\operatorname*{Op}\nolimits_{\tau}(a\circ m_{L})\psi|\phi)_{L^{2}}%
\end{align*}
hence the equality (\ref{aml}).
\end{proof}

\subsection{Cohen class property}

\subsubsection{Definition of the Cohen class}

Let $Q:\mathcal{S}(\mathbb{R}^{n})\times\mathcal{S}(\mathbb{R}^{n}%
)\longrightarrow\mathcal{S}(\mathbb{R}^{2n})$ be a sesquilinear form and set
$Q\psi=Q(\psi,\psi)$. Recall that $Q\psi$ belongs to the Cohen class if it is
of the type $Q\psi=W\psi\ast\theta$ for some distribution $\theta
\in\mathcal{S}^{\prime}(\mathbb{R}^{2n})$. Sufficient conditions for a
distribution to belong to Cohen's class are%
\begin{gather}
Q\psi(z-z_{0})=Q(\widehat{T}(z_{0})\psi)(z)\label{cohen1}\\
|Q(\psi,\phi)(0,0)|\leq C||\psi||_{L^{2}}||\phi||_{L^{2}}\text{ }
\label{cohen2}%
\end{gather}
where $C$ is a constant (see e.g. Gr\"{o}chenig \cite{Gro} or de Gosson
\cite{Birkbis}). Taking $Q(\psi,\phi)=W_{\tau}(\psi,\phi)$ condition
(\ref{cohen1}) is easily seen to hold, but condition (\ref{cohen2}) only holds
when $\tau\neq0$ and $\tau\neq1$; in fact a straightforward calculation using
the Cauchy--Schwarz inequality yields the estimate
\begin{equation}
|W_{\tau}(\psi,\phi)(0)|\leq\left(  \tfrac{1}{2\pi\hbar}\right)  ^{N}\tfrac
{1}{\tau^{N/2}(1-\tau)^{N/2}}||\psi||_{L^{2}}||\phi||_{L^{2}}. \label{cotau}%
\end{equation}
Boggiatto et al. \cite{bogetal} however show by a direct calculation that when
$\hbar=1/2\pi$ one has%
\begin{equation}
W_{\tau}(\psi,\phi)=W(\psi,\phi)\ast\alpha_{\tau} \label{bog1}%
\end{equation}
with
\[
\alpha_{\tau}(z)=\left(  \tfrac{2}{|2\tau-1|}\right)  ^{N}e^{2\pi i\tfrac
{2}{2\tau-1}px}.
\]
when $\tau\neq\frac{1}{2}$. It follows that:

\begin{proposition}
\label{propos4}For $\psi,\phi\in\mathcal{S}(\mathbb{R}^{N})$ we have
\begin{equation}
W_{\tau}(\psi,\phi)=W(\psi,\phi)\ast\theta_{\tau} \label{bog12}%
\end{equation}
where%
\begin{equation}
\theta_{\tau}(z)=\left(  \tfrac{1}{|2\tau-1|\pi\hbar}\right)  ^{N}e^{\frac
{i}{\hbar}\tfrac{2}{2\tau-1}px} \label{bog2}%
\end{equation}
when $\tau\neq\frac{1}{2}$. When $\tau=\frac{1}{2}$ we have $\theta_{\tau
}=\delta$.
\end{proposition}

\begin{proof}
Let us denote $W_{\tau}^{2\pi}(\psi,\phi)$ the transform $W_{\tau}(\psi,\phi)$
when $\hbar=1/2\pi$; it is related to the general case by the obvious formula%
\begin{equation}
W_{\tau}^{2\pi}(\psi,\phi)(x,p)=\left(  2\pi\hbar\right)  ^{N}W_{\tau}%
(\psi,\phi)(x,2\pi\hbar); \label{toipie}%
\end{equation}
the result immediately follows from (\ref{bog1}) using elementary changes of
variables. The case $\tau=\frac{1}{2}$ is straightforward since $W_{1/2}%
(\psi,\phi)=W(\psi,\phi)$.
\end{proof}

\subsubsection{Applications to $\protect\widehat{A}_{\tau}$}

We are going to establish an important representation result for $\tau
$-pseudo-differential operators using the Heisenberg operator. Let us first
prove the following Lemma which is a straightforward consequence of the
Proposition above:

\begin{lemma}
\label{lem1}Let $\widehat{A}_{\tau}=\operatorname*{Op}\nolimits_{\tau}(a)$
with $a\in\mathcal{S}^{\prime}(\mathbb{R}^{2N})$. The Weyl symbol
$a_{\mathrm{W}}$ of $\widehat{A}_{\tau}$ is given by $a_{\mathrm{W}}%
=a\ast\theta_{\tau}$.
\end{lemma}

\begin{proof}
In view of the equality (\ref{aweyl}) the Weyl symbol of $\widehat{A}_{\tau}$
is determined by formula (\ref{aftau}):%
\[
(\widehat{A}_{\tau}\psi|\phi)_{L^{2}}=\langle a_{\mathrm{W}},W(\psi
,\phi)\rangle
\]
for $\psi,\phi\in\mathcal{S}(\mathbb{R}^{N})$. In view of (\ref{awtau}) we
also have
\[
(\widehat{A}_{\tau}\psi|\phi)_{L^{2}}=\langle a,W_{\tau}(\psi,\phi)\rangle
\]
that is, taking (\ref{bog12}) into account,%
\[
(\widehat{A}_{\tau}\psi|\phi)_{L^{2}}=\langle a,W(\psi,\phi)\ast\theta_{\tau
}\rangle=\langle a\ast\theta_{\tau}^{\vee},W(\psi,\phi)\rangle
\]
where $\theta_{\tau}^{\vee}(z)=\theta_{\tau}(-z)$ hence $a_{\mathrm{W}}%
=a\ast\theta_{\tau}^{\vee}$; since $\theta_{\tau}(-z)=\theta_{\tau}(z)$ we
have $a_{\mathrm{W}}=a\ast\theta_{\tau}$ as claimed.
\end{proof}

\begin{proposition}
\label{prop3}The action of a pseudo-differential operator $\widehat{A}_{\tau}$
on $\psi\in\mathcal{S}(\mathbb{R}^{N})$ is given by%
\begin{equation}
\widehat{A}_{\tau}\psi=\left(  \tfrac{1}{2\pi\hbar}\right)  ^{N}%
\int_{\mathbb{R}^{2N}}\mathcal{F}_{\sigma}a(z_{0})\widehat{T}_{\tau}%
(z_{0})\psi dz_{0} \label{atw1}%
\end{equation}
where
\begin{equation}
\mathcal{F}_{\sigma}a(z)=\left(  \tfrac{1}{2\pi\hbar}\right)  ^{N}%
\int_{\mathbb{R}^{2N}}e^{-\frac{i}{\hbar}\sigma(z,z^{\prime})}a(z^{\prime
})dz^{\prime}=\mathcal{F}a(Jz) \label{fis}%
\end{equation}
is the symplectic Fourier transform of $a\in\mathcal{S}^{\prime}%
(\mathbb{R}^{2n})$ and $\widehat{T}_{\tau}(z_{0})$ is the modified Heisenberg
operator defined by%
\begin{equation}
\widehat{T}_{\tau}(z_{0})\psi(x)=e^{\frac{i}{2\hbar}(2\tau-1)p_{0}x_{0}%
}\widehat{T}(z_{0})\psi(x) \label{hopt}%
\end{equation}
that is
\begin{equation}
\widehat{T}_{\tau}(z_{0})\psi(x)=e^{\frac{i}{\hbar}(p_{0}x-(1-\tau)p_{0}%
x_{0})}\psi(x-x_{0}). \label{hoptbis}%
\end{equation}

\end{proposition}

\begin{proof}
Assume that $\tau=\frac{1}{2}$, then $\widehat{T}_{\tau}(z_{0})=\widehat{T}%
(z_{0})$ and formula (\ref{atw1}) becomes%
\[
\widehat{A}\psi=\left(  \tfrac{1}{2\pi\hbar}\right)  ^{N}\int_{\mathbb{R}%
^{2N}}\mathcal{F}_{\sigma}a(z_{0})\widehat{T}(z_{0})\psi dz_{0}%
\]
which is the expression of a Weyl operator well-known in harmonic analysis
(see e.g. de Gosson \cite{Birk,Birkbis}). When $\tau\neq\frac{1}{2}$ we argue
as follows: in view of Lemma \ref{lem1} we have%
\begin{align*}
\widehat{A}_{\tau}\psi &  =\left(  \tfrac{1}{2\pi\hbar}\right)  ^{N}%
\int_{\mathbb{R}^{2N}}\mathcal{F}_{\sigma}(a\ast\theta_{\tau})(z_{0}%
)\widehat{T}(z_{0})\psi dz_{0}\\
&  =\int_{\mathbb{R}^{2N}}\mathcal{F}_{\sigma}a(z_{0})\mathcal{F}_{\sigma
}\theta_{\tau}(z_{0})\widehat{T}(z_{0})\psi dz_{0}%
\end{align*}
where we have used the formula $\mathcal{F}_{\sigma}(a\ast\theta_{\tau
})=\left(  2\pi\hbar\right)  ^{N}\mathcal{F}_{\sigma}a\mathcal{F}_{\sigma
}\theta_{\tau}$. which follows at once from the usual formula giving the
Fourier transform of a convolution product. A straightforward calculation
shows that we have%
\[
\mathcal{F}_{\sigma}\theta_{\tau}(z_{0})=\left(  \tfrac{1}{2\pi\hbar}\right)
^{N}e^{\frac{i}{2\hbar}(2\tau-1)p_{0}x_{0}}%
\]
and hence%
\[
\widehat{A}_{\tau}\psi=\left(  \tfrac{1}{2\pi\hbar}\right)  ^{N}%
\int_{\mathbb{R}^{2N}}\mathcal{F}_{\sigma}a(z_{0})e^{\frac{i}{2\hbar}%
(2\tau-1)p_{0}x_{0}}\widehat{T}(z_{0})\psi dz_{0}%
\]
which is precisely formula (\ref{atw1}).
\end{proof}

An easy calculation shows that the adjoint of the operator $\widehat{T}_{\tau
}(z_{0})$ is given by
\begin{equation}
\widehat{T}_{\tau}(z_{0})^{\ast}=\widehat{T}_{1-\tau}(-z_{0}); \label{tetauz}%
\end{equation}
from this immediately follows the formula for the adjoint of $\widehat{A}%
_{\tau}=\operatorname*{Op}_{\tau}(a)$:
\begin{equation}
\operatorname*{Op}\nolimits_{\tau}(a)^{\ast}=\operatorname*{Op}%
\nolimits_{1-\tau}(\overline{a}). \label{tetad}%
\end{equation}

\begin{remark}
A straightforward computation using either (\ref{hopt}) or (\ref{hoptbis})
shows that the modified Heisenberg operators $\widehat{T}_{\tau}(z_{0})$
satisfy for all values of $\tau$ the commutation relations
\begin{equation}
\widehat{T}_{\tau}(z_{0})\widehat{T}_{\tau}(z_{1})=e^{\frac{i}{\hbar}%
\sigma(z_{0},z_{1})}\widehat{T}_{\tau}(z_{1})\widehat{T}_{\tau}(z_{0}).
\label{comm}%
\end{equation}
These operators thus correspond to (equivalent) representations of the
Heisenberg group.
\end{remark}

\subsubsection{The $\tau$-dependent cross-ambiguity transform}

The cross-Wigner transform $W(\psi,\phi)$ has a \textquotedblleft dual
companion\textquotedblright, the cross-ambiguity transform $A(\psi,\phi)$
which is explicitly given by the integral formula%
\begin{equation}
A(\psi,\phi)(z)=\left(  \tfrac{1}{2\pi\hbar}\right)  ^{N}\int_{\mathbb{R}^{N}%
}e^{-\tfrac{i}{\hbar}px^{\prime}}\psi(x^{\prime}+\tfrac{1}{2}x)\overline
{\phi(x^{\prime}-\tfrac{1}{2}x)}dx^{\prime}. \label{fouwig}%
\end{equation}
One toggles between both using the symplectic Fourier transform:%
\begin{equation}
A(\psi,\phi)=F_{\sigma}W(\psi,\phi)\text{ \ , \ }W(\psi,\phi)=F_{\sigma}%
A(\psi,\phi). \label{toggle}%
\end{equation}
We are going to generalize this formula to the $\tau$-dependent case; let us
first recall the following alternative definition of the cross-ambiguity
transform (see de Gosson \cite{Birk,Birkbis}):
\begin{equation}
A(\psi,\phi)(z)=\left(  \tfrac{1}{2\pi\hbar}\right)  ^{n}(\psi|\widehat{T}%
(z)\phi)_{L^{2}}. \label{ambigbis}%
\end{equation}
This formula suggests that we define%
\begin{equation}
A_{\tau}(\psi,\phi)(z)=\left(  \tfrac{1}{2\pi\hbar}\right)  ^{n}%
(\psi|\widehat{T}_{\tau}(z)\phi)_{L^{2}} \label{ambigtau}%
\end{equation}
where $\widehat{T}_{\tau}(z)$ is the modified Heisenberg operator
(\ref{hopt}). A straightforward calculation gives the explicit expression%
\begin{multline}
A_{\tau}(\psi,\phi)(z_{0})=\left(  \tfrac{1}{2\pi\hbar}\right)  ^{N}%
e^{-\frac{i}{2\hbar}(2\tau-1)p_{0}x_{0}}\\
\times\int_{\mathbb{R}^{N}}e^{-\tfrac{i}{\hbar}p_{0}x^{\prime}}\psi(x^{\prime
}+\tau x_{0})\overline{\phi(x^{\prime}-(1-\tau)x_{0})}dx^{\prime}.\nonumber
\end{multline}

\section{Born--Jordan Quantization\label{sec3}}

In what follows the parameter $\tau$ is restricted to the closed interval
$[0,1]$.

\subsection{The Born--Jordan operators $\protect\widehat{A}_{\mathrm{BJ}}$}

\subsubsection{Definition of $\protect\widehat{A}_{\mathrm{BJ}}$}

Let $\widehat{A}_{\tau}=\operatorname*{Op}_{\tau}(a)$ be the
pseudo-differential operator defined by formula (\ref{aftau}). By definition
the Born--Jordan operator $\widehat{A}_{\mathrm{BJ}}=\operatorname*{Op}%
_{\mathrm{BJ}}(a)$ is the average of the operators $\widehat{A}_{\tau}$ for
$\tau\in\lbrack0,1]$:
\begin{equation}
\widehat{A}_{\mathrm{BJ}}\psi=\left(  \tfrac{1}{2\pi\hbar}\right)  ^{N}%
\int_{0}^{1}\widehat{A}_{\tau}\psi d\tau. \label{abj}%
\end{equation}
Note that it immediately follows from formula (\ref{adjtau}) for the adjoint
of $\widehat{A}_{\tau}$ that we have%
\begin{equation}
\operatorname*{Op}\nolimits_{\mathrm{BJ}}(a)^{\ast}=\operatorname*{Op}%
\nolimits_{\mathrm{BJ}}(\overline{a}) \label{opabj}%
\end{equation}
hence, in particular, $\widehat{A}_{\mathrm{BJ}}=\operatorname*{Op}%
\nolimits_{\mathrm{BJ}}(a)$ is (formally) self-adjoint if and only if the
symbol $a$ is real. This important property is thus common to Born--Jordan and
Weyl calculus, and makes $\widehat{A}_{\mathrm{BJ}}$ a good candidate for a
physical quantization procedure. But more about that later.

To justify the chosen terminology we have to show that the quantization
$a\longrightarrow\operatorname*{Op}_{\mathrm{BJ}}(a)$ contains as a particular
case the original Born--Jordan prescription (\ref{bj1}) described in the
Introduction. That is we have to prove that
\begin{equation}
\operatorname*{Op}\nolimits_{\mathrm{BJ}}(x_{j}^{m}p_{j}^{n})=\frac{1}%
{n+1}\sum_{k=0}^{n}\widehat{p}_{j}^{n-k}\widehat{x}_{j}^{m}\widehat{p}_{j}%
^{k}. \label{bj3}%
\end{equation}
Recall that we have shown in Proposition \ref{propos5} (formula (\ref{txp2}))
that%
\[
\operatorname*{Op}\nolimits_{\tau}(x_{j}^{m}p_{j}^{n})=\sum_{k=0}^{n}%
\begin{pmatrix}
n\\
k
\end{pmatrix}
(1-\tau)^{k}\tau^{n-k}\widehat{p_{j}}^{k}\widehat{x_{j}}^{n}\widehat{p_{j}%
}^{n-k}.
\]
It follows that
\[
\operatorname*{Op}\nolimits_{\mathrm{BJ}}(x_{j}^{m}p_{j}^{n})=\sum_{k=0}^{n}%
\begin{pmatrix}
n\\
k
\end{pmatrix}
B(n-k+1,k+1)\widehat{p_{j}}^{k}\widehat{x_{j}}^{n}\widehat{p_{j}}^{m-k}%
\]
where $B$ is the beta function. Since%
\begin{align*}
B(k+1,n-k+1)  &  =\frac{\Gamma(k+1)\Gamma(n-k+1)}{\Gamma(n+2)}\\
&  =\frac{k!(n-k)!}{(n+1)!}%
\end{align*}
we have%
\[
\operatorname*{Op}\nolimits_{\mathrm{BJ}}(x_{j}^{m}p_{j}^{n})=\frac{1}%
{n+1}\sum_{k=0}^{m}\widehat{p_{j}}^{k}\widehat{x_{j}}^{n}\widehat{p_{j}}^{n-k}%
\]
which is the same thing as (\ref{bj3}). Notice that if we had started with
formula (\ref{txp1}) instead of the equivalent to formula (\ref{txp2}) the
same argument yields the alternative equality%
\begin{equation}
\operatorname*{Op}\nolimits_{\mathrm{BJ}}(x_{j}^{m}p_{j}^{n})=\frac{1}%
{m+1}\sum_{k=0}^{m}\widehat{x_{j}}^{k}\widehat{p_{j}}^{n}\widehat{x_{j}}%
^{m-k}. \label{bj3bis}%
\end{equation}

\subsubsection{Comparison of Born--Jordan and Weyl quantization}

A quadratic Hamiltonian
\begin{equation}
H(z)=\frac{1}{2}Mz^{2}=(x,p)M(x,p)^{T} \label{quad}%
\end{equation}
where $M=M^{T}$ is a real $2N\times2N$ matrix has identical Weyl and
Born--Jordan quantizations; in fact writing the Hamiltonian as
\[
H(z)=\sum_{j}\alpha_{j}p_{j}^{2}+\beta_{j}x_{j}^{2}+2\gamma_{j}p_{j}x_{j}%
\]
we see that $\operatorname*{Op}\nolimits_{\mathrm{BJ}}(H)=\operatorname*{Op}%
(H)$ when the $\gamma_{j}$ are all zero; when there are cross-terms
$x_{j}p_{j}$ the claim follows using formula (\ref{bj3}) (or (\ref{bj3bis}))
with $m=n=1$; this shows that the Born--Jordan quantization of $x_{j}p_{j}$ is
$\frac{1}{2}(\widehat{x_{j}}\widehat{p_{j}}+\widehat{p_{j}}\widehat{x_{j}})$
which is the same result as that obtained using Weyl quantization (cf. formula
(\ref{w2})). In both case the corresponding operator is thus given by%
\[
\widehat{H}=\frac{1}{2}(\widehat{x_{j}},\widehat{p_{j}})M(\widehat{x_{j}%
},\widehat{p_{j}})^{T}.
\]

Born--Jordan and Weyl quantization are also identical for \textquotedblleft
physical\textquotedblright\ Hamiltonians of the type \textquotedblleft%
\textit{kinetic energy} $+$ \textit{potential}\textquotedblright. If $H$ is a
symbol of the type (\ref{h1}) that is%
\begin{equation}
H=\sum_{j=1}^{N}\frac{1}{2m_{j}}p_{j}^{2}+V(x) \label{h4}%
\end{equation}
then $\widehat{H}=\operatorname*{Op}\nolimits_{\mathrm{BJ}}%
(H)=\operatorname*{Op}(H)$ is given by%
\begin{equation}
\widehat{H}=\sum_{j=1}^{N}\frac{-\hbar^{2}}{2m_{j}}\frac{\partial^{2}%
}{\partial x_{j}^{2}}+V(x). \label{h5}%
\end{equation}
This can be seen by noting that $\operatorname*{Op}\nolimits_{\mathrm{BJ}%
}(p_{j}^{2})=-\hbar^{2}\partial^{2}/\partial x_{j}^{2}$ \ taking $m=0$ and
$n=2$ in formula (\ref{bj3bis}) and then using definition (\ref{aftau}):
\begin{align*}
\operatorname*{Op}\nolimits_{\tau}(V)\psi(x)  &  =\left(  \tfrac{1}{2\pi\hbar
}\right)  ^{N}\int_{\mathbb{R}^{2N}}e^{\frac{i}{\hbar}p\cdot(x-y)}V(\tau
x+(1-\tau)y)\psi(y)dydp\\
&  =\int_{\mathbb{R}^{N}}V(\tau x+(1-\tau)y)\psi(y)\delta(x-y)dy\\
&  =V(x)\psi(x);
\end{align*}
integrating in $\tau$ from $0$ to $1$ yields $\operatorname*{Op}%
\nolimits_{\mathrm{BJ}}(V)\psi=V\psi$ and hence (\ref{h5}).

More generally the Born--Jordan and Weyl quantizations of the magnetic
Hamiltonian (\ref{h3}) also coincide; let us first prove the following useful Lemma:

\begin{lemma}
Let $\mathcal{A}:\mathbb{R}^{N}\times\mathbb{R}_{t}\longrightarrow\mathbb{R}$
be a smooth function. Then
\begin{equation}
\operatorname*{Op}\nolimits_{\mathrm{BJ}}(p_{j}\mathcal{A})\psi
=\operatorname*{Op}(p_{j}\mathcal{A})\psi=-\frac{i\hbar}{2}\left[
\frac{\partial}{\partial x}(\mathcal{A\psi)}+\mathcal{A}\frac{\partial
}{\partial x}\psi\right]  . \label{pwbj}%
\end{equation}

\end{lemma}

\begin{proof}
It is sufficient to assume $N=1$. Using definition (\ref{aftau}) of
$\widehat{A}_{\tau}=\operatorname*{Op}_{\tau}(a)$ we have
\begin{align*}
\operatorname*{Op}\nolimits_{\tau}(p\mathcal{A})\psi(x)  &  =\tfrac{1}%
{2\pi\hbar}\int_{\mathbb{R}^{2}}e^{\frac{i}{\hbar}p\cdot(x-y)}p\mathcal{A}%
(\tau x+(1-\tau)y,t)\psi(y)dydp\\
&  =\int_{-\infty}^{\infty}\left[  \tfrac{1}{2\pi\hbar}\int_{-\infty}^{\infty
}e^{\frac{i}{\hbar}p\cdot(x-y)}pdp\right]  \mathcal{A}(\tau x+(1-\tau
)y,t)\psi(y)dy.
\end{align*}
In view of formula (\ref{fif}) the expression between the square brackets is
$-i\hbar\delta^{\prime}(x-y)$ hence%
\begin{align*}
\operatorname*{Op}\nolimits_{\tau}(p\mathcal{A})\psi(x)  &  =-i\hbar
\int_{-\infty}^{\infty}\delta^{\prime}(x-y)\mathcal{A}(\tau x+(1-\tau
)y,t)\psi(y)dy\\
&  =-i\hbar\int_{-\infty}^{\infty}\delta(x-y)\tfrac{\partial}{\partial
y}\left[  \mathcal{A}(\tau x+(1-\tau)y,t)\psi(y)\right]  dy\\
&  =-i\hbar\left[  (1-\tau)\tfrac{\partial}{\partial x}(\mathcal{A}\psi
)+\tau\mathcal{A}\tfrac{\partial}{\partial x}\psi\right]  .
\end{align*}
Formula (\ref{pwbj}) follows setting in the Weyl case $\tau=\frac{1}{2}$ and
integrating from $0$ to $1$ in the Born--Jordan case.
\end{proof}

It follows from the Lemma above that both Weyl and Born--Jordan quantizations
of a (time-dependent) magnetic Hamiltonian
\begin{equation}
H(z,t)=\sum_{j=1}^{N}\frac{1}{2m_{j}}\left(  p_{j}-\mathcal{A}_{j}%
(x,t)\right)  ^{2}+V(x,t) \label{magnetic}%
\end{equation}
are the same. In fact, expanding the terms $\left(  p_{j}-\mathcal{A}%
_{j}(x,t)\right)  ^{2}$ we get%
\[
H=\sum_{j=1}^{N}\frac{1}{2m_{j}}p_{j}^{2}-\sum_{j=1}^{N}\frac{1}{m_{j}}%
p_{j}\mathcal{A}_{j}+\sum_{j=1}^{N}\mathcal{A}_{j}{}^{2}+V.
\]
We have seen above that the terms $p_{j}^{2}$ and $\mathcal{A}_{j}{}^{2}+V$
have identical quantizations; in view of formula (\ref{pwbj}) this also true
of the cross-terms $p_{j}\mathcal{A}_{j}$, leading in both cases to the
expression%
\begin{equation}
\widehat{H}=\sum_{j=1}^{N}\frac{1}{2m_{j}}\left(  -i\hbar\frac{\partial
}{\partial x_{j}}-\mathcal{A}_{j}(x,t)\right)  ^{2}+V(x,t) \label{magquant}%
\end{equation}
well-known from standard quantum mechanics.

\section{Some Properties of Born--Jordan Quantization}

\subsubsection{Harmonic representation of $\protect\widehat{A}_{\mathrm{BJ}}$}

It is customary in harmonic analysis to write a Weyl operator $\widehat{A}$
with symbol $a$ in the form%
\begin{equation}
\widehat{A}\psi=\left(  \tfrac{1}{2\pi\hbar}\right)  ^{N}\int_{\mathbb{R}%
^{2N}}\mathcal{F}a(x_{0},p_{0})e^{\frac{i}{\hbar}(\widehat{x}x_{0}%
+\widehat{p}p_{0})}\psi dp_{0}dx_{0}; \label{ah1}%
\end{equation}
this formula goes back to the work of Weyl \cite{Weyl}. It is however
preferable for our study of Born--Jordan quantization to use the alternative
formulation%
\begin{equation}
\widehat{A}\psi=\left(  \tfrac{1}{2\pi\hbar}\right)  ^{N}\int_{\mathbb{R}%
^{2N}}\mathcal{F}_{\sigma}a(z_{0})\widehat{T}(z_{0})\psi dz_{0} \label{ah2}%
\end{equation}
already used in the proof of Proposition \ref{prop3}. This not only because
the role of the Heisenberg group in this procedure becomes more apparent, but
also because practical calculations are easier and more explicit. The
equivalence of both formulas is clear (at least at the formal level):
replacing $z_{0}=(x_{0},p_{0})$ in (\ref{ah1}) with $Jz_{0}$ one gets%
\[
\widehat{A}\psi=\left(  \tfrac{1}{2\pi\hbar}\right)  ^{N}\int_{\mathbb{R}%
^{2N}}\mathcal{F}a(Jz_{0})e^{\frac{i}{\hbar}\sigma(\widehat{z},z_{0})}\psi
dp_{0}dx_{0}%
\]
which is precisely (\ref{ah2}) since $\widehat{T}(z_{0})=e^{\frac{i}{\hbar
}\sigma(\widehat{z},z_{0})}$.

\begin{proposition}
\label{propharmonic}Let $\psi\in\mathcal{S}(\mathbb{R}^{N})$. Following
properties hold:

(i) We have
\begin{equation}
\widehat{A}_{\mathrm{BJ}}\psi=\left(  \tfrac{1}{2\pi\hbar}\right)  ^{N}%
\int_{\mathbb{R}^{2N}}\mathcal{F}_{\sigma}a(z_{0})\Theta(z_{0})\widehat{T}%
(z_{0})\psi dz_{0}\label{abj1}%
\end{equation}
where $\Theta$ is the function defined by%
\begin{equation}
\Theta(z_{0})=\frac{\sin(p_{0}x_{0}/\hbar)}{p_{0}x_{0}/\hbar}.\label{tbj1}%
\end{equation}

(ii) The Weyl symbol $a_{\mathrm{W}}$ of $\widehat{A}_{\mathrm{BJ}}$ is given
by the convolution product
\begin{equation}
a_{\mathrm{W}}=\left(  \tfrac{1}{2\pi\hbar}\right)  ^{N}a\ast\mathcal{F}%
_{\sigma}\Theta.\label{aw}%
\end{equation}

\end{proposition}

\begin{proof}
(i) In view of formulas (\ref{atw1}) and (\ref{hopt}) in Proposition
\ref{prop3} we have%
\[
\widehat{A}_{\mathrm{BJ}}\psi=\left(  \tfrac{1}{2\pi\hbar}\right)  ^{N}%
\int_{\mathbb{R}^{2N}}\mathcal{F}_{\sigma}a(z_{0})\left(  \int_{0}^{1}%
e^{\frac{i}{2\hbar}(2\tau-1)p_{0}x_{0}}d\tau\right)  \widehat{T}(z_{0})\psi
dz_{0};
\]
a straightforward calculation yields%
\[
\int_{0}^{1}e^{\frac{i}{2\hbar}(2\tau-1)p_{0}x_{0}}d\tau=\frac{2\hbar}%
{p_{0}x_{0}}\sin\frac{p_{0}x_{0}}{2\hbar}%
\]
hence formula (\ref{abj1}).

(ii) Since the symplectic Fourier transform is involutive we have
$(\mathcal{F}_{\sigma}a)\Theta=(\mathcal{F}_{\sigma}a)\mathcal{F}_{\sigma
}(\mathcal{F}_{\sigma}\Theta)$ hence (\ref{aw}) since%
\[
(\mathcal{F}_{\sigma}a)\mathcal{F}_{\sigma}(\mathcal{F}_{\sigma}%
\Theta)=\left(  \tfrac{1}{2\pi\hbar}\right)  ^{N}\mathcal{F}_{\sigma}%
(a\ast\mathcal{F}_{\sigma}\Theta).
\]

\end{proof}

In \cite{bogetal} Boggiatto et al. consider the average
\[
W_{\mathrm{BJ}}(\psi,\phi)(z)=\int_{0}^{1}W_{\tau}(\psi,\phi)(z)dt
\]
(which they denote by $Q(\psi,\phi)$); they show that the bilinear form
$W_{\mathrm{BJ}}$ belongs to the Cohen class. It immediately follows from
(\ref{marginal1}) that the marginal properties also hold for $W_{\mathrm{BJ}%
}\psi=W_{\mathrm{BJ}}(\psi,\phi)$:%
\begin{equation}
\int_{\mathbb{R}^{N}}W_{\mathrm{BJ}}\psi(z)dp=|\psi(x)|^{2}\text{ ,\ }%
\int_{\mathbb{R}^{N}}W_{\mathrm{BJ}}\psi(z)dx=|\mathcal{F}\psi(p)|^{2}%
.\label{marginal2}%
\end{equation}

As expected, Born--Jordan operators can be expressed in terms of their symbol
and $W_{\mathrm{BJ}}(\psi,\phi)$:

\begin{proposition}
The operator $\widehat{A}_{\mathrm{BJ}}$ and the bilinear form $W_{\mathrm{BJ}%
}$ are related by the formula%
\begin{equation}
(\widehat{A}_{\mathrm{BJ}}\psi|\phi)_{L^{2}}=\langle a,W_{\mathrm{BJ}}%
(\psi,\phi)\rangle\label{awbj}%
\end{equation}
valid for all $\psi,\phi\in\mathcal{S}(\mathbb{R}^{N})$.
\end{proposition}

\begin{proof}
In view of formula (\ref{awtau}) in Proposition \ref{propos1} we have%
\[
(\widehat{A}_{\tau}\psi|\phi)_{L^{2}}=\langle a,W_{\tau}(\psi,\phi)\rangle;
\]
integrating this equality from $0$ to $1$ with respect to the variable $\tau$
yields (\ref{awbj}).
\end{proof}

\subsubsection{Symplectic covariance of $\protect\widehat{A}_{\mathrm{BJ}}$}

Since a Born--Jordan operator is in general distinct from the Weyl operator
with same symbol we cannot expect full symplectic covariance to hold for them. However:

\begin{proposition}
Let $\widehat{A}_{\mathrm{BJ}}=\operatorname*{Op}_{\mathrm{BJ}}(a)$. We have:

(i) Let $F\in\operatorname*{Mp}(2n,\mathbb{R})$ be the modified Fourier
transform $i^{-d/2}\mathcal{F}$; then%
\begin{equation}
F^{-1}\operatorname*{Op}\nolimits_{\mathrm{BJ}}(a)F=\operatorname*{Op}%
\nolimits_{\mathrm{BJ}}(a\circ J); \label{fobj}%
\end{equation}

(ii) Let $M_{L,\mu}\in\operatorname*{ML}(2n,\mathbb{R})$\ (the metalinear
group) and $m_{L}=\pi^{\operatorname*{Mp}}(M_{L,m})\in\operatorname*{Sp}%
(2N,\mathbb{R})$; we have%
\begin{equation}
M_{L,\mu}^{-1}\widehat{A}_{\mathrm{BJ}}M_{L,\mu}=\operatorname*{Op}%
\nolimits_{\mathrm{BJ}}(a\circ m_{L}). \label{mobl}%
\end{equation}

\end{proposition}

\begin{proof}
(i) In view of formula (\ref{optauf}) we have
\[
F^{-1}\operatorname*{Op}\nolimits_{\tau}(a)F=\operatorname*{Op}%
\nolimits_{1-\tau}(a\circ J)
\]
hence%
\[
F^{-1}\left(  \int_{0}^{1}\operatorname*{Op}\nolimits_{\tau}(a)d\tau\right)
F=\int_{0}^{1}\operatorname*{Op}\nolimits_{1-\tau}(a\circ J)d\tau=\int_{0}%
^{1}\operatorname*{Op}\nolimits_{\tau}(a\circ J)d\tau
\]
and formula (\ref{fobj}) follows. (ii) In view of formula (\ref{aml}) in
Proposition \ref{propos2} we have%
\[
M_{L,\mu}^{-1}\left(  \int_{0}^{1}\operatorname*{Op}\nolimits_{\tau}%
(a)d\tau\right)  M_{L,\mu}=\int_{0}^{1}\operatorname*{Op}\nolimits_{\tau
}(a\circ m_{L})d\tau
\]
hence the covariance formula (\ref{mobl}).
\end{proof}

\begin{remark}
It is possible to give a direct proof of (\ref{fobj}) and (\ref{mobl}) using
the explicit formula (\ref{abj1}) for $\widehat{A}_{\mathrm{BJ}}\psi$, the
symplectic covariance of Weyl operators, and the fact that the function
$\Theta$ given by (\ref{tbj1}) is invariant under the transformations
$(x,p)\longmapsto(p,-x)$ and $(x,p)\longmapsto(L^{-1}x,L^{T}p)$.
\end{remark}

Recalling that
\[
M_{L,\mu}\psi(x)=i^{\mu}\sqrt{|\det L|}\psi(Lx)
\]
the operators $F$ and $M_{L,\mu}$ satisfy the intertwining formula%
\[
FM_{L,\mu}=M_{(L^{T})^{-1},\mu}F,
\]
hence the set $\{F,M_{L,\mu}:\det L\neq0\}$ is a subgroup of the metaplectic
group $\operatorname*{Mp}(2n,\mathbb{R})$. The result above says that the
Born--Jordan operators are covariant under the action of this group.

\section*{Discussion}

We have seen that for physical Hamiltonians of the type
\[
H=\sum_{j=1}^{N}\frac{1}{2m_{j}}\left(  p_{j}-\mathcal{A}_{j}(x)\right)
^{2}+V(x)
\]
both Weyl and Born--Jordan quantizations are the same, and so are the
quantizations of the generalized harmonic oscillator (\ref{quad}). One could
therefore wonder whether it is really worth to bother and study the
differences between both quantization schemes. The reason might come from the
fact that Weyl quantization is in a sense \textquotedblleft too
perfect\textquotedblright. It is, as Kauffmann \cite{kauffmann} points out,
the most \textquotedblleft austere\textquotedblright\ quantization, and this
austerity enables it to have very good symmetry properties. In particular it
has the property of symplectic covariance, and it is the only
pseudo-differential calculus having this feature, as follows from the argument
in Wong \cite{Wong}. As we briefly mentioned in the Introduction Weyl
correspondence $a\longmapsto\operatorname*{Op}_{\mathrm{Weyl}}(a)$ is
invertible, and establishes a bijection between symbols $a\in\mathcal{S}%
^{\prime}(\mathbb{R}^{2N})$ and continuous operators $\widehat{A}%
:\mathcal{S}(\mathbb{R}^{N})\longrightarrow\mathcal{S}^{\prime}(\mathbb{R}%
^{N})$. This allows (see de Gosson and Hiley \cite{gohi}) to show that
conceptually speaking Schr\"{o}dinger's equation is equivalent to Hamilton's
equations of motion. Such a situation is not physically tenable (unless one
introduces supplementary interpretational condition justifying the
introduction of Planck's constant), because quantum and Hamiltonian mechanics
are certainly not equivalent theories (at least physically)! It turns out that
Born--Jordan quantization is not invertible. This question of
\textquotedblleft dequantization\textquotedblright\ is very important, and
perhaps more important than that of \textquotedblleft
dequantization\textquotedblright\ as was already stressed by Mackey
\cite{Mackey}. Let us shortly discuss the (non)invertibility of the
Born--Jordan correspondence $a\longmapsto\operatorname*{Op}_{\mathrm{BJ}}(a)$.
Let $\widehat{A}:\mathcal{S}(\mathbb{R}^{N})\longrightarrow\mathcal{S}%
^{\prime}(\mathbb{R}^{N})$ be a an arbitrary continuous linear operator with
Weyl symbol $a_{\mathrm{W}}\in\mathcal{S}^{\prime}(\mathbb{R}^{2N})$ be its :
$\widehat{A}=\operatorname*{Op}_{\mathrm{Weyl}}(a_{\mathrm{W}})$. If there
exists $a\in\mathcal{S}^{\prime}(\mathbb{R}^{2N})$ such that $\widehat{A}%
=\operatorname*{Op}_{\mathrm{BJ}}(a)$ then in view of formula (\ref{aw}) in
Proposition \ref{propharmonic} $a_{\mathrm{W}}$ and $a$ are related by the
convolution equation
\[
a_{\mathrm{W}}=\left(  \tfrac{1}{2\pi\hbar}\right)  ^{N}a\ast\mathcal{F}%
_{\sigma}\Theta
\]
that is, taking (symplectic) Fourier transforms
\[
\mathcal{F}_{\sigma}a_{\mathrm{W}}=(\mathcal{F}_{\sigma}a)\Theta
\]
However, given an arbitrary $a_{\mathrm{W}}\in\mathcal{S}^{\prime}%
(\mathbb{R}^{2N})$ this relation does not determine $\mathcal{F}_{\sigma}a$,
that is $a.$ This fact, together with the properties of the distribution
$Q(\psi,\phi)$ studied by Boggiatto et al. \cite{bogetal} suggests that
Born--Jordan quantization could really make a case against more traditional
quantization schemes. This possibility should certainly be studied seriously,
and perhaps complemented using recent results in Boggiatto et al
\cite{bogetalbis}) where the authors consider weighted averages of the
quasi-distributions $W_{\tau}$. We add that Molahajloo \cite{shahla} has
recently considered the $\tau$-quantization of Laplacian operators in
connection with a study of the heat kernel; it would probably be interesting
to investigate the corresponding Born--Jordan quantization.

The study of quantization, both from mathematical and physical perspectives,
is certainly not closed and still has a brilliant future!

\begin{acknowledgement}
The authors would like to express their gratitude to Professor S. Kauffmann
for several exchanges of ideas on (de)quantization and quantization, and for
an interesting discussion of invertibility of quantization.
\end{acknowledgement}

\end{document}